\documentclass[12pt]{article}

\usepackage{indentfirst}
\usepackage{amsmath}
\usepackage{amsfonts}
\usepackage{mathrsfs}
\usepackage{amsthm}
\usepackage{nicefrac}
\usepackage[shortlabels]{enumitem}
\usepackage{graphicx}
\usepackage[labelformat=simple]{subcaption}

\newtheorem{theorem}{Theorem}
\newtheorem{conjecture}{Conjecture}
\theoremstyle{definition}
\newtheorem*{definition}{Definition}
\newtheorem{example}{Example}
\newtheorem*{remark}{Remark}

\DeclareMathOperator{\End}{End}

\begin{document}

\title{The spectrum of the Berezin transform for Gelfand pairs}

\author{\textsc Dor Shmoish$^{1}$}

\footnotetext[1]{Partially supported by the European Research Council Advanced grant 338809}

\date{\today}

\maketitle

\begin{abstract}
We discuss the Berezin transform, a Markov operator associated to positive-operator valued measures (POVMs). We consider the class of so-called orbit POVMs, constructed on the quotient space $\Omega = G/K$ of a compact group $G$ by its subgroup $K$. We restrict attention to the case where $(G, K)$ is a Gelfand pair and derive an explicit formula for the spectrum of the Berezin transform in terms of the characters of the irreducible unitary representations of $G$. We then specialize our results to the case study $G = \text{SU}(2)$ and $K \simeq S^1$, and find the spectra of orbit POVMs on $S^2$. We confirm previous calculations by Zhang and Donaldson of the spectrum of the standard quantization of $S^2$ coming from Kähler geometry. Then, we make a couple of conjectures about the oscillations in the sequence of eigenvalues, and prove them in the simplest case of second-highest weight vector. Finally, for low weights, we prove that the corresponding orbit POVMs on $S^2$ violate the axioms of a Berezin-Toeplitz quantization.
\end{abstract}

\tableofcontents

\section{Introduction and Main Results} \label{sec:1}
The main subject of the present paper is mathematical quantization. In classical mechanics, the phase space is a symplectic manifold $M$ and observables are modeled by the space $C^\infty(M)$ of smooth functions on $M$, whereas in quantum mechanics, the phase space is a complex Hilbert space $\mathcal{H}$ and observables are modeled by the space $\mathscr{L}(\mathcal{H})$ of Hermitian operators on $\mathcal{H}$. Quantum states are provided by so-called density operators, which are positive trace-one operators forming a subset $\mathcal{S}(\mathcal{H}) \subset \mathscr{L}(\mathcal{H})$. Quantization is the procedure of constructing a quantum system starting from the classical mechanics of a system, in such a way that the classical system is the limit as $\hbar \rightarrow 0$ of the quantum system. Here $\hbar$ is Planck's constant, which in this setting is just a parameter of the construction. Since the goal of quantization is to find a quantum system that is analogous in some sense to a given classical system, there is no unique approach to it.\par

Here we focus on the Berezin-Toeplitz quantization procedure, introduced for the first time by Berezin in \cite{Berezin}. In fact, we restrict our attention to the Berezin-Toeplitz quantization of closed Kähler manifolds \cite{Berezin, Bordemann, Schlichenmaier, LeFloch}. Such a quantization is defined by a sequence of positive surjective linear maps $T_\hbar: C^\infty(M) \rightarrow \mathscr{L}(\mathcal{H}_\hbar)$  with $T_\hbar(1) = \text{Id}$. The sequence is parametrized by $\hbar \in \Lambda$ for some subset $\Lambda \subset \mathbb{R}$ having $0$ as a limit point. That $T_\hbar$ is positive means that for any $f \ge 0$ we have $T_\hbar(f) \ge 0$. These maps have to satisfy the following properties:
\begin{enumerate}[label=(\arabic*), topsep=3pt]
    \item \textbf{(norm correspondence)} $\|f\|_\infty - O(\hbar) \le \| T_\hbar(f) \|_\text{op} \le \|f\|_\infty$
    \item \textbf{(bracket correspondence)} $\left| -\frac{i}{\hbar} [T_\hbar(f), T_\hbar(g)] - T_\hbar(\{f, g\}) \right|_\text{op} = O(\hbar)$
    \item \textbf{(quasi-multiplicativity)} $\| T_\hbar(fg) - T_\hbar(f) T_\hbar(g) \|_\text{op} = O(\hbar)$
    \item \textbf{(trace correspondence)} $\left| \text{tr}\,\big(T_\hbar(f)\big) - \frac{1}{(2\pi\hbar)^n} \int_M f\,d\text{Vol} \right| = O\big(\frac{1}{\hbar^{n-1}}\big)$
\end{enumerate}
for all $f, g \in C^\infty(M)$ and all $\hbar \in \Lambda$, where $n = \dim_\mathbb{C} M$.\par

Every positive linear map $T: L_2(M) \rightarrow \mathscr{L}(\mathcal{H})$ satisfying $T(1) = \text{Id}$ is given by integration with respect to some POVM on $M$. A POVM (positive-operator valued measure) on $M$ is, roughly speaking, a mapping associating a positive-definite operator to any measurable subset of $M$ in a $\sigma$-additive manner. Formally, if $M$ is equipped with the $\sigma$-algebra $\mathcal{C}$, then an $\mathscr{L}(\mathcal{H})$-valued POVM on $(M, \mathcal{C})$ is a mapping $W: \mathcal{C} \rightarrow \mathscr{L}(\mathcal{H})$ that takes each subset $X \in \mathcal{C}$ to a positive operator $W(X) \in \mathscr{L}(\mathcal{H})$ in a countably additive manner, normalized by $W(M) = \text{Id}$. Given a positive linear map $T: L_2(M) \rightarrow \mathscr{L}(\mathcal{H})$ satisfying $T(1) = \text{Id}$, define $W: \mathcal{C} \rightarrow \mathscr{L}(\mathcal{H})$ by the equality $W(X) = T(1_X)$ for every $X \in \mathcal{C}$. Then we indeed obtain an $\mathscr{L}(\mathcal{H})$-valued POVM on $(M, \mathcal{C})$, and we clearly have $T(\phi) = \int_M \phi \, dW$ for every $\phi \in L_2(M)$. This is because we have equality whenever $\phi$ is an indicator function by definition, and both sides of the equality are countably additive.\par

It is known \cite{CDS} that an $\mathscr{L}(\mathcal{H})$-valued POVM $W$ on $(M, \mathcal{C})$ has a density with respect to some probability measure $\alpha$ on $(M, \mathcal{C})$, i.e. has the form
\begin{equation*}
dW(x) = n \, F(x) \, d\alpha(x),
\end{equation*}
where $n = \dim \mathcal{H}$ and $F: M \rightarrow \mathcal{S}(\mathcal{H})$ is a measurable function.\par

Thus every quantization map $T_\hbar: C^\infty(M) \rightarrow \mathscr{L}(\mathcal{H}_\hbar)$ extended to $L_2(M)$ by continuity is given by integration with respect to an $\mathscr{L}(\mathcal{H}_\hbar)$-valued POVM $W_\hbar$, which has the form $dW_\hbar(x) = n_\hbar \, F_\hbar(x) \, d\alpha_\hbar(x)$. Incidentally, integrating with respect to the measure $\alpha_\hbar$ instead of the standard volume form, the trace correspondence principle can be stated as a precise equality:
\begin{enumerate}[label=(\arabic*'), topsep=3pt]
  \setcounter{enumi}{3}
  \item \textbf{(trace correspondence)} $\text{tr}\,\big(T_\hbar(f)\big) = n_\hbar \int_M f\,d\alpha_\hbar$.
\end{enumerate}\par

Given a quantization scheme, we may consider the following operation. For a function $f$ on the classical phase space $M$, let us first quantize it and then dequantize. Quantization is performed by applying the map $T_\hbar$, while dequantization is performed by applying the dual map $T_\hbar^*$. We again obtain a function on the phase space $M$, which is a blurring of the original function $f$. This operation on functions, $f \mapsto \mathcal{B}_\hbar f$, is called the Berezin transform. Formally, the Berezin transform is defined by the equation $\mathcal{B}_\hbar := \frac{1}{n_\hbar} T_\hbar^* T_\hbar$, where $n_\hbar = \dim \mathcal{H}_\hbar$.\par

One can generalize the definition of the Berezin transform and define it given any POVM on $M$. For an $\mathscr{L}(\mathcal{H})$-valued POVM $W$ on $M$, which has the form $dW(x) = n \, F(x) \, d\alpha(x)$, the corresponding quantization map $T: L_2(M) \rightarrow \mathscr{L}(\mathcal{H})$ is given by
\begin{equation*}
T(\phi) = \int_M \phi \, dW = n \int_M \phi(x) F(x) \, d\alpha(x).
\end{equation*}
The dequantization map $T^*$ is the dual mapping of $T$ with respect to the inner products $\left< \phi, \psi \right> = \int_M \phi \, \overline{\psi} \, d\alpha$ on $C^\infty(M)$ and $\left< A, B \right> = \text{tr}(AB)$ on $\mathscr{L}(\mathcal{H})$. This map $T^*: \mathscr{L}(\mathcal{H}) \rightarrow L_2(M)$ is given by
\begin{equation*}
T^*(A)(x) = n \: \text{tr}(F(x)A).
\end{equation*}
Finally, the Berezin transform is defined by $\mathcal{B} := \frac{1}{n} T^* T$.\par

The Berezin transform naturally arises in two different settings: in the context of the Berezin-Toeplitz quantization of closed Kähler manifolds, and when considering certain POVMs associated to irreducible representations of finite or compact groups \cite{Kaminker}.\par

In the framework of the Berezin-Toeplitz quantization of closed Kähler manifolds, $\mathcal{B}_\hbar$ is known to be a Markov operator with finite-dimensional image. We focus on the spectral properties of $\mathcal{B}_\hbar$. For fixed $\hbar$, this operator factors through a finite-dimensional space and hence its spectrum consists of a finite collection of points lying in the interval $[0,1]$. Moreover, multiplicities of positive eigenvalues are finite, and $1$ is the maximal eigenvalue corresponding to the constant function. Write its spectrum (with multiplicities) in the form
\begin{equation*}
1 = \lambda_0 \ge \lambda_1 \ge \hdots \ge \lambda_k \ge \hdots \ge 0.
\end{equation*}
The quantity $\gamma := 1 - \lambda_1$ is called the spectral gap, a fundamental characteristic of a Markov chain responsible for the rate of convergence to the stationary distribution.\par

In addition to quantization, POVMs appear in quantum mechanics in another setting: they model quantum measurements \cite{Busch}. Interestingly enough, within this model the spectral gap of the Berezin transform corresponding to a POVM admits two different interpretations: it measures the minimal magnitude of quantum noise production, and it equals the spectral gap of the Markov chain corresponding to repeated quantum measurements.\par

In the present paper, we study the spectral properties of the Berezin transform of a certain class of so-called orbit POVMs, whose construction was briefly described in the Preliminaries and Remark 6.7 of \cite{IKPS}. It is a representation-theoretic construction following ideas first introduced by Perelomov in \cite[p.~223]{Perelomov}.\par

Given a compact group $G$ with normalized Haar measure $\mu$, fix an irreducible unitary representation $(\rho, V)$ of $G$ and a vector $v \in V$. Consider the subgroup
\begin{equation*}
K = \left\{ g \in G \mid \rho(g)v = e^{i\theta}v \text{ for some } \theta \in [0,2\pi) \right\}
\end{equation*}
of elements whose action on $v$ merely changes its phase. Thinking of vectors in $V$ as pure quantum states, vectors that differ only in phase correspond to the same state, and thus $K$ can be thought of as the stabilizer of $v$.\\
Now consider the orbit space $\Omega = G/K$ equipped with the pushforward measure $\omega(X) = \mu(XK)$, and define an $\mathscr{L}(V)$-valued POVM $W$ on $\Omega$ by
\begin{equation*}
dW(x) = n \, P_{\rho(\tilde{x})v} \: d\omega(x),
\end{equation*}
where $n = \dim V$, $\tilde{x}$ is any lifting of $x$ to $G$ and $P_{\rho(\tilde{x})v}$ is the orthogonal projection onto the vector $\rho(\tilde{x})v$. Note that $W$ is well-defined: if $g_1 K = g_2 K$, then there is some $k \in K$ with $g_2 = g_1 k$, and then
\begin{equation*}
P_{\rho(g_2) v} = P_{\rho(g_1) \rho(k) v} = P_{\rho(g_1) e^{i\theta} v} = P_{e^{i\theta} \rho(g_1) v} = P_{\rho(g_1) v}
\end{equation*}
(elements in the same coset differ in their action on $v$ only by a phase factor, which doesn't affect the projection operator).\\
We shall refer to a POVM obtained by this construction as an orbit POVM.\par

$W$ gives rise to a quantization map $T: L_2(\Omega) \rightarrow \mathscr{L}(V)$ given by
\begin{equation*}
T(f) = \int_\Omega f \, dW = n \int_\Omega f(x) P_{\rho(\tilde{x})v} \, d\omega(x),
\end{equation*}
with the dual map $T^*: \mathscr{L}(V) \rightarrow L_2(\Omega)$ given by
\begin{equation*}
T^*(A)(x) = n \: \text{tr} \left( P_{\rho(\tilde{x})v} A \right).
\end{equation*}
Thus we have the associated Berezin transform $\mathcal{B}: L_2(\Omega, \omega) \rightarrow L_2(\Omega, \omega)$,
\begin{equation*}
\mathcal{B} = \frac{1}{n} T^* T,
\end{equation*}
and our goal in the present work is to study its spectrum.\par

Let us recall some definitions first.

\begin{definition}[Fourier transform]
Let $G$ be a compact group with normalized Haar measure $\mu$ and denote by $\widehat{G}$ the set of equivalence classes of unitary irreducible representations of $G$. Let $f \in L_2(G, \mu)$ be any square-integrable complex-valued function. Then the Fourier transform of $f$ by an irreducible representation $\varphi \in \widehat{G}$ is the operator $\widehat{f}(\varphi) \in \End V_\varphi$ defined by
\begin{equation*}
\widehat{f}(\varphi) = \int\limits_G f(x) \varphi(x^{-1}) d\mu(x).
\end{equation*}
\end{definition}

\begin{definition}[Gelfand pair]
Let $G$ be a compact group and let $K$ be a subgroup. Let $C(G)$ be the convolution algebra of continuous complex-valued functions on $G$, and let $C^\#_K(G)$ be the subalgebra of bi-$K$-invariant functions, i.e. functions $f \in C(G)$ satisfying $f(k_1 g k_2) = f(g)$ for all $k_1, k_2 \in K$ and $g \in G$. The pair $(G, K)$ is said to be a Gelfand pair if the convolution algebra $C^\#_K(G)$ is commutative (cf. Definition 6.1.1 in \cite{vanDijk}).
\end{definition}

\begin{example}
Let $G$ be a compact abelian group, and let $K = \{e\}$. Then $(G, K)$ is a Gelfand pair.
\end{example}

\begin{example}
Let $G = \text{SO}(n)$. The group $G$ acts transitively on the sphere $S^{n-1} \subset \mathbb{R}^n$. Let $K$ be the stabilizer of the point $e_1 = (1, 0, \hdots, 0)$, so that $K \simeq \text{SO}(n-1)$. Then $(G, K)$ is a Gelfand pair \cite[p.~95]{vanDijk}.
\end{example}

Let us now introduce some common notation.\\
For a representation $\varphi \in \widehat{G}$ denote by $\chi_\varphi(x) = \text{tr}\left(\varphi(x)\right)$ the character of $\varphi$.\\
For $f, g \in L_2(G, \mu)$ denote by $\left< f, g \right> = \int_G f(x) \, \overline{g(x)} \, d\mu(x)$ their standard inner product on $L_2(G, \mu)$.\par
Our main result can be described as follows. Let $W$ be an orbit POVM on $\Omega = G/K$, constructed by fixing a representation $\rho \in \widehat{G}$ and a vector $v \in V_\rho$. Define the function $u: G \rightarrow \mathbb{R}$ by
\begin{equation*}
u(g) = n \left| \left< \rho(g)v, v \right> \right|^2.
\end{equation*}
Assuming that $(G, K)$ is a Gelfand pair, the spectrum of the associated Berezin transform $\mathcal{B}$ consists of the values of the coefficients $\left< u, \chi_\varphi \right>$ in the expansion of the class function $u$ by the basis $(\chi_\varphi)_{\varphi \in \widehat{G}}$ of characters of the irreducible unitary representations of $G$:

\begin{theorem} \label{thm:1}
Let $W$ be an orbit POVM on $\Omega = G/K$ defined by the equality $dW(x) = n \, P_{\rho(\tilde{x})v} \, d\omega$, and assume that $(G, K)$ is a Gelfand pair. Let $\mathcal{B}$ denote the associated Berezin transform and define the function $u:G \rightarrow \mathbb{R}$ by $u(g) = n \left| \left< \rho(g)v, v \right> \right|^2$. Then
\begin{equation*}
\mathrm{Spec}(\mathcal{B}) = \left\{ \left< u, \chi_\varphi \right> \mathrel{\big|} \varphi \in \widehat{G} \right\} \cup \{0\},
\end{equation*}
where each $\left< u, \chi_\varphi \right>$ has multiplicity $\dim \varphi$ (and $0$ has infinite multiplicity).
\end{theorem}

To prove this result, we begin with a study of the spectrum of the Berezin transform of general orbit POVMs (without the assumption that $(G, K)$ is a Gelfand pair). We first discover that the Berezin transform of an orbit POVM is a convolution operator which acts on functions by convolution with the function $u$ defined above. Then, via harmonic analysis and some linear algebra, we obtain the expression
\begin{equation*}
\text{Spec}^*(\mathcal{B}) = \bigsqcup_{\varphi \in \widehat{G}} \left(\, \bigsqcup_{j=1}^{\dim \varphi} \text{Spec}^*\big(\widehat{u}(\varphi)\big) \right),
\end{equation*}
where $\text{Spec}^*$ is the multiset of all eigenvalues repeated according to their multiplicity. We then restrict attention to the case where $(G, K)$ is a Gelfand pair, since in this case there is a particularly simple expression for $\widehat{u}(\varphi)$. This expression allows us to easily derive the result of Theorem \ref{thm:1}.\par

We then focus on the case $G = \text{SU(2)}$ and
\begin{equation*}
K = \left\{ k_t := \begin{pmatrix} e^{it} & 0\\0 & e^{-it} \end{pmatrix} \mathrel{\Big|} t \in [0, 2\pi) \right\} \simeq S^1,
\end{equation*}
and hence consider orbit POVMs on the phase space $G/K \simeq S^2$. Recall that the irreducible unitary representations of $\text{SU}(2)$ are given by $\left\{ \rho_j \mid j \in \frac{1}{2} \mathbb{N} \right\}$ where $\rho_j$ is a representation of dimension $n_j := 2j+1$ (Theorem 5.6.3 in \cite{Kowalski}), whose space $V_j$ has an orthonormal basis $w_j^{(j)}, w_j^{(j-1)}, ..., w_j^{(-j)}$ consisting of eigenvectors with respect to $K$, $\rho_j(k_t) w_j^{(m)} = e^{2imt} w_j^{(m)}$. The parameter $m$ in $w_j^{(m)}$ is called the weight of the vector. We fix $j \in \frac{1}{2} \mathbb{N}$ and take $\rho = \rho_j$ and $v = w_j^{(m)}$. We thus have the POVM $dW_{j, m} = {n_j \, P_{j, m} \, d\omega}$ on $S^2$, where $P_{j, m}([g]) = P_{\rho_j(g)w_j^{(m)}}$ is the orthogonal projection onto the vector $\rho_j(g)w_j^{(m)}$, with associated Berezin transform $\mathcal{B}_{j, m}$ and the corresponding function
\begin{equation*}
u_{j, m}(g) = n_j \left|\left< \rho_j(g) w_j^{(m)}, w_j^{(m)} \right>\right|^2.
\end{equation*}
We begin by noting that in this case $(G, K)$ is a Gelfand pair, and hence Theorem \ref{thm:1} applies. We proceed with an explicit calculation of the values $\left< u_{j, m}, \chi_\varphi \right>$ using tools from representation theory and well-known formulas for the Clebsch-Gordan coefficients to obtain the spectrum of $\mathcal{B}_{j, m}$ explicitly:

\begin{theorem} \label{thm:2}
The spectrum of the Berezin transform $\mathcal{B}_{j,m}$ is given by
\begin{equation*}
\mathrm{Spec}(\mathcal{B}_{j,m}) = \{ \lambda^{(0)}, \lambda^{(1)}, \hdots, \lambda^{(2j)}, 0\},
\end{equation*}
where $\lambda^{(J)} := \frac{(2j)!(2j+1)!}{(2j-J)!(2j+J+1)!} \left( \sum\limits_{z=0}^{j-m} (-1)^z \frac{\binom{2j-J}{z}\binom{J}{j-m-z}^2}{\binom{2j}{j-m}} \right)^2$ has multiplicity \break $2J+1$ (and $0$ has infinite multiplicity).
\end{theorem}

An important corollary of this result is a new proof of the formula for the spectrum of the Berezin transform of the orbit POVM obtained by choosing the highest weight vector. In this case, the POVMs $W_{j, j}$ give rise to the quantization maps $T_j(f) := \int_{S^2} f \, dW_{j, j}$ which provide a quantization scheme that is known to be equivalent to the standard quantization of $S^2$ coming from Kähler geometry. This follows from the fact that the coherent states in both cases are the same (cf. eq. (43) in \cite{Perelomov} and Definition 5.1.1, Example 7.1.8 and Theorem 7.2.1 in \cite{LeFloch}) In this case, Theorem \ref{thm:2} tells us that
\begin{equation*}
\lambda^{(J)} = \frac{(2j)!(2j+1)!}{(2j-J)!(2j+J+1)!} = \frac{2j \cdot \hdots \cdot (2j-J+1)}{(2j+J+1) \cdot \hdots \cdot (2j+2)},
\end{equation*}
in agreement with prior calculations by Zhang \cite[p.~385]{Zhang} and by Donaldson \cite[p.~613]{Donaldson}. It is easily verified that these eigenvalues satisfy
\begin{equation*}
1 = \lambda^{(0)} > \lambda^{(1)} > \lambda^{(2)} > \hdots > \lambda^{(2j)} > 0,
\end{equation*}
and hence the spectral gap is
\begin{equation*}
\gamma(\mathcal{B}_{j, j}) = 1 - \lambda^{(1)} = 1 - \frac{j}{j+1} = \frac{1}{j+1} = 2\hbar + O(\hbar^2).
\end{equation*}\par

We then turn our attention to lower weight vectors. We first consider the case $m = j-d$, where $d \in \mathbb{N}^+$ is a constant. In this case, it is no longer true that $\lambda^{(0)} > \lambda^{(1)} > \lambda^{(2)} > \hdots > \lambda^{(2j)}$.\\
To gain some insight, we choose $j = 100$ and for $j-m \in \{1,2,3,4\}$ we plot the values of the eigenvalues $\lambda^{(0)}, \lambda^{(1)}, ..., \lambda^{(200)}$.

\begin{figure}[ht!]
  \centering
  \begin{subfigure}[t]{0.48\linewidth}
    \includegraphics[width=\linewidth]{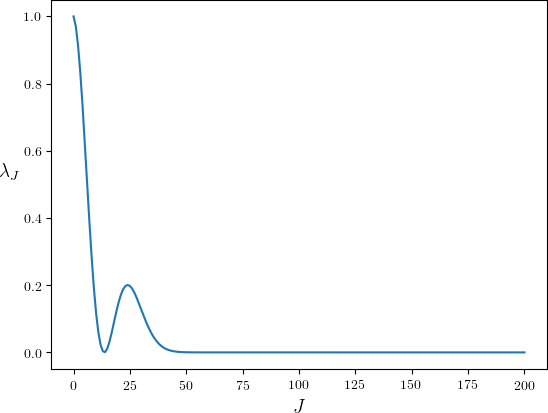}
     \caption{$j-m = 1$.}
  	 \label{fig:1a}
  \end{subfigure}
  \hfill
  \begin{subfigure}[t]{0.48\linewidth}
    \includegraphics[width=\linewidth]{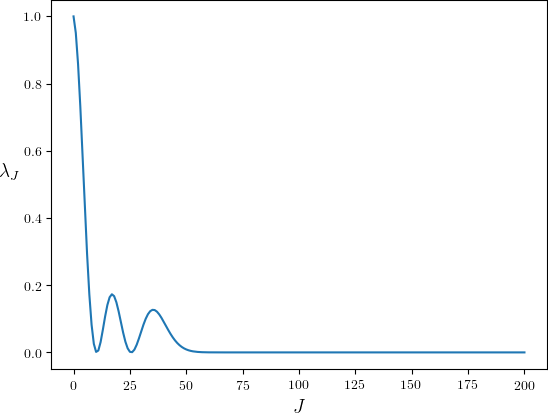}
    \caption{$j-m = 2$.}
  \end{subfigure}
  \par \smallskip \smallskip \smallskip
  \begin{subfigure}[t]{0.48\linewidth}
    \includegraphics[width=\linewidth]{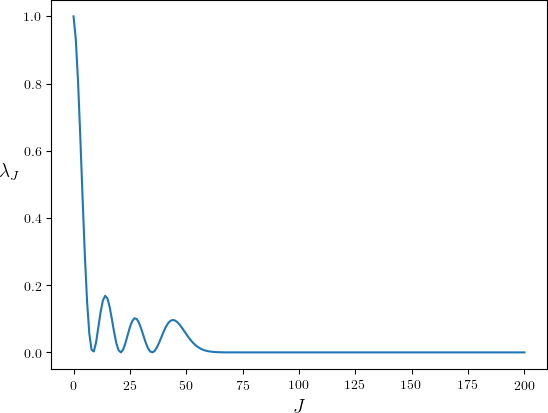}
    \caption{$j-m = 3$.}
  \end{subfigure}
  \hfill
  \begin{subfigure}[t]{0.48\linewidth}
    \includegraphics[width=\linewidth]{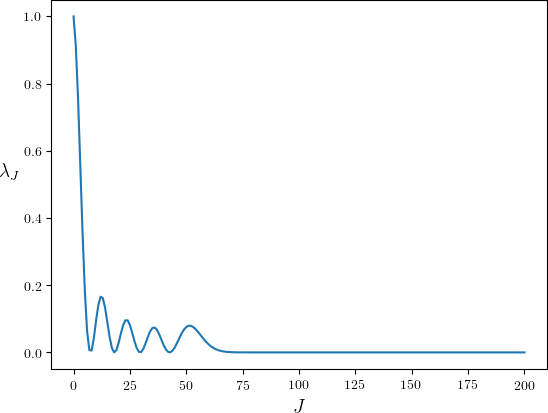}
    \caption{$j-m = 4$.}
  \end{subfigure}
  \caption{The spectrum of $\mathcal{B}_{j,m}$ for $j = 100$ and $j-m \in \{1,2,3,4\}$.}
  \label{fig:1}
\end{figure}

\noindent Looking at the plots shown in Figure \ref{fig:1} leads us to the following conjectures.

\begin{conjecture} \label{conj:1}
For every $d \in \mathbb{N}$ there exists $j_0$ such that for all $j \ge j_0$ and $m = j-d$, the eigenvalues of the Berezin transform $\mathcal{B}_{j, j-d}$ satisfy
\begin{equation*}
\lambda^{(1)} > \lambda^{(2)}, \lambda^{(3)}, \hdots, \lambda^{(2j)},
\end{equation*}
and hence in particular the spectral gap is
\begin{equation*}
\gamma(\mathcal{B}_{j, j-d}) = 1 - \lambda^{(1)} = 1 - \frac{(j-d)^2}{j(j+1)} = \frac{2d+1}{j} - \frac{(d+1)^2}{j(j+1)}.
\end{equation*}
\end{conjecture}

\begin{conjecture} \label{conj:2}
For every $d \in \mathbb{N}$ there exists $j_0$ such that for all $j \ge j_0$ and $m = j-d$ the sequence $\lambda^{(0)}, \lambda^{(1)}, \lambda^{(2)}, \hdots, \lambda^{(2j)}$ of eigenvalues of the Berezin transform $\mathcal{B}_{j, j-d}$ has $d$ local minima and $d$ local maxima.
\end{conjecture}

We then prove these conjectures in the simplest case ($d = 1$) of second-highest weight vector using straightforward algebra and estimates:

\begin{theorem} \label{thm:3}
For all $j \in \frac{1}{2} \mathbb{N}$, the eigenvalues of the Berezin transform $\mathcal{B}_{j, j-1}$ satisfy
\begin{equation*}
\lambda^{(0)} > \lambda^{(1)} > \hdots > \lambda^{(\lfloor\sqrt{2j}\rfloor)} < \hdots < \lambda^{(\lfloor\sqrt{6j}\rfloor)} > \hdots > \lambda^{(2j)}.
\end{equation*}
Moreover, for $j \ge \frac{5}{2}$, we have $\lambda^{(1)} > \lambda^{(\lfloor\sqrt{6j}\rfloor)}$, so that
\begin{equation*}
\lambda^{(1)} > \lambda^{(2)}, \lambda^{(3)}, \hdots, \lambda^{(2j)},
\end{equation*}
and hence, in particular,
\begin{equation*}
\gamma(\mathcal{B}_{j, j-1}) = 1 - \lambda^{(1)} = 1 - \frac{(j-1)^2}{j(j+1)} = \frac{3}{j} - \frac{4}{j(j+1)} = 6\hbar + O(\hbar^2).
\end{equation*}
\end{theorem}

Finally, we consider the case where $j - |m|$ is unbounded. We then again have a sequence of POVMs $\left(W_{j, m_j}\right)_{j \in \frac{1}{2} \mathbb{N}}$, and the first question one should be concerned with is whether such a sequence of POVMs yields a Berezin-Toeplitz quantization. It turns out that the answer is negative.

\begin{theorem} \label{thm:4}
Let $\left(m_j\right)_{j \in \frac{1}{2} \mathbb{N}}$ be a sequence such that $m_j \in \{ -j, -j+1, \hdots, j \}$ and assume $(j-|m_j|)_{j \in \frac{1}{2} \mathbb{N}}$ is unbounded. Consider the sequence of POVMs $\left(W_{j, m_j}\right)_{j \in \frac{1}{2} \mathbb{N}}$ and let $\Lambda = (1/k)_{k \in \mathbb{N}^+}$. Then the sequence of maps $\left(Q_\hbar\right)_{\hbar \in \Lambda}$ defined by $Q_\hbar(f) = \int_{S^2} f \, dW_{j, m_j}$, where $j$ and $\hbar$ are connected via $\hbar = \frac{1}{2j}$, does not satisfy the properties of a Berezin-Toeplitz quantization.
\end{theorem}

To prove this theorem, we use the formula for the spectrum of $\mathcal{B}_{j, m}$ given by Theorem \ref{thm:2} in conjunction with a classification of certain quantizations of $S^2$ obtained in \cite{IKP}.\par

The analogous question for the case where $j-m = d$ is a positive constant remains open, even for the second-highest weight.\par

The rest of the paper is organized as follows.\par

In Section \ref{sec:2} we start with the calculation of the spectrum of the Berezin transform of an orbit POVM in terms of the spectra of the Fourier coefficients of the associated function $u$ introduced above. Then we focus on the case of Gelfand pairs and prove Theorem \ref{thm:1}.\par

In Section \ref{sec:3} we restrict attention to the particular case where $G = \text{SU}(2)$ and $K \simeq S^1$, and carry out a fully explicit calculation of the spectra of Berezin transforms of orbit POVMs on $G/K \simeq S^2$ obtained by choosing vectors of arbitrary weights, thus proving Theorem \ref{thm:2}. Then we use these expressions for the spectra in order to study orbit POVMs obtained from non-highest weight vectors and prove Theorems \ref{thm:3} and \ref{thm:4}.

\section{Spectrum of the Berezin Transform of an Orbit POVM} \label{sec:2}
Let $G$ be a compact group with normalized Haar measure $\mu$, so that $\mu(G) = 1$. Fix a non-trivial unitary irreducible representation $\rho: G \rightarrow U(V)$, of dimension $\dim V = n$, and fix a vector $v \in V$. Define the subgroup
\begin{equation*}
K = \left\{ g \in G \mid \rho(g)v = e^{i\theta}v \text{ for some } \theta \in [0,2\pi) \right\}
\end{equation*}
and let $W$ be the $\mathscr{L}(V)$-valued orbit POVM on the orbit space $\Omega = G/K$ defined by
\begin{equation*}
dW(x) = n \, P_{\rho(\tilde{x})v} \: d\omega(x).
\end{equation*}
We then have the corresponding map $T: L_2(\Omega, \omega) \rightarrow \mathscr{L}(V)$,
\begin{equation*}
T(f) = \int_\Omega f \, dW = n \int_\Omega f(x) P_{\rho(\tilde{x})v} \, d\omega(x),
\end{equation*}
and the dual map $T^*: \mathscr{L}(V) \rightarrow L_2(\Omega, \omega)$,
\begin{equation*}
T^*(A)(x) = n \: \text{tr} \left( P_{\rho(\tilde{x})v} A \right).
\end{equation*}
We thus have the associated Berezin transform $\mathcal{B}: L_2(\Omega, \omega) \rightarrow L_2(\Omega, \omega)$,
\begin{equation*}
\mathcal{B} = \frac{1}{n} T^* T.
\end{equation*}
Our goal in the following two subsections is to arrive at the expression (\ref{eq:7}) in order to study the spectrum of $\mathcal{B}$ and prove Theorem \ref{thm:1}.

\subsection{The Berezin Transform as a Convolution Operator}
Explicitly, the operator $\mathcal{B}$ is given by
\begin{equation*}
\begin{split}
(\mathcal{B}f)(s) & = n \int\limits_\Omega f(t) \, \text{tr} \left( P_{\rho(t)v} P_{\rho(s)v} \right) d\omega(t) \\
 & = \int\limits_\Omega n \: \text{tr} \left( P_{\rho(s)v} P_{\rho(t)v} \right) f(t) \, d\omega(t),
\end{split}
\end{equation*}
and thus has kernel
\begin{equation*}
\mathcal{B}(s,t) = n \: \text{tr} \left( P_{\rho(s)v} P_{\rho(t)v} \right).
\end{equation*}
Since
\begin{equation*}
\text{tr} \left( P_{\rho(s)v} P_{\rho(t)v} \right) = \left| \left<\rho(s)v, \rho(t)v\right>\right|^2 = \left| \left< \rho(t)^{-1} \rho(s) \, v, v \right> \right|^2 = \left| \left< \rho(t^{-1}s) \, v, v \right> \right|^2
\end{equation*}
(recall $\rho$ is unitary), we have
\begin{equation*}
\mathcal{B}(s,t) = u(t^{-1}s)
\end{equation*}
where the function $u: G \rightarrow \mathbb{R}$ is defined by
\begin{equation*}
u(g) = n \left| \left< \rho(g)v, v \right> \right|^2.
\end{equation*}
It can be readily verified that $u$ is a bi-$K$-invariant function.\par
For a function $f \in L_2(\Omega, \omega)$, let $F(g) := f([g])$ be its lifting to $G$, so that $F \in L_2(G, \mu)$ and $F$ is right-$K$-invariant. Then
\begin{equation*}
\begin{split}
(\mathcal{B}f)(s) & = \int\limits_{\Omega} \mathcal{B}(s,t) \, f(t) \, d \omega(t) = \int\limits_{\Omega} u(t^{-1}s) \, f(t) \, d \omega(t) \\
 & = \int\limits_{\Omega} \int\limits_{K} F(tk) \, u \big( (tk)^{-1} s \big) \, dk \, d \omega(t) = \int\limits_{G} F(t) \, u(t^{-1}s) \, d \mu(t)
\end{split}
\end{equation*}
It follows that
\begin{equation*}
\mathcal{B}f = \left. \left( F * u \right) \right|_\Omega.
\end{equation*}

\subsection{The Eigenfunction Equation}
Assume $f \in L_2(\Omega, \omega)$ is an eigenfunction of $\mathcal{B}$ with eigenvalue $\lambda$,
\begin{equation} \label{eq:1}
\mathcal{B}f = \lambda f.
\end{equation}
Then we have the equality
\begin{equation*}
\left. \left( F * u \right) \right|_\Omega = \lambda f,
\end{equation*}
which can be lifted to $G$ by right-$K$-invariance as
\begin{equation} \label{eq:2}
F * u = \lambda F.
\end{equation}
Conversely, if (\ref{eq:2}) holds, and $\lambda \neq 0$, then the right-$K$-invariance of $u$ implies the right-$K$-invariance of $F * u$, therefore $F = \frac{1}{\lambda} (F * u)$ is right-$K$-invariant too,
and hence $f = \left. F \right|_\Omega$ satisfies (\ref{eq:1}).\\
We thus obtain for any $\lambda \neq 0$ a one-to-one correspondence between $\lambda$-eigenfunctions $f$ of $\mathcal{B}$ and functions $F \in L_2(G, \mu)$ satisfying (\ref{eq:2}).\par
To further investigate (\ref{eq:2}) we invoke Theorem 5.5.7 from \cite{Kowalski}, which states that for any $\zeta \in L_2(G, \mu)$,
\begin{equation} \label{eq:3}
\zeta(x) = \sum_{\varphi \in \widehat{G}} (\dim{\varphi}) \: \text{tr}\left( \widehat{\zeta}(\varphi) \varphi(x) \right),
\end{equation}
where
\begin{equation} \label{eq:4}
\widehat{\zeta}(\varphi) = \int\limits_G \zeta(x) \varphi(x^{-1}) d\mu(x).
\end{equation}
We conclude that (\ref{eq:2}) is equivalent to
\begin{equation} \label{eq:5}
\widehat{F*u}(\varphi) = \widehat{\lambda F}(\varphi)
\end{equation}
for all $\varphi \in \widehat{G}$.
It is a standard fact that
\begin{equation*}
\widehat{F*u}(\varphi) = \widehat{u}(\varphi) \widehat{F}(\varphi),
\end{equation*}
and thus (\ref{eq:5}) reduces to
\begin{equation} \label{eq:6}
\widehat{u}(\varphi) \widehat{F}(\varphi) = \lambda \widehat{F}(\varphi)
\end{equation}
for all $\varphi \in \widehat{G}$.\par
Note that (\ref{eq:3}), (\ref{eq:4}) give a linear bijection
\begin{equation*}
\Psi: L_2(G, \mu) \,\xrightarrow{\sim}\, \prod_{\varphi \in \widehat{G}} \End V_\varphi,
\end{equation*}
where $V_\varphi$ is the finite-dimensional vector space of the representation $\varphi$, of dimension $d_\varphi := \dim \varphi$.
Now (\ref{eq:6}) implies that $F \in L_2(G, \mu)$ satisfies (\ref{eq:2}) if and only if $\Psi(F)$ satisfies $T_u \Psi(F) = \lambda \Psi(F)$, where
\begin{equation*}
T_u: \prod_{\varphi \in \widehat{G}} \End V_\varphi \rightarrow \prod_{\varphi \in \widehat{G}} \End V_\varphi,
\end{equation*}
\begin{equation*}
T_u \left( (M_\varphi)_{\varphi \in \widehat{G}} \right) = (\widehat{u}(\varphi) M_\varphi)_{\varphi \in \widehat{G}}.
\end{equation*}
We conclude that
\begin{equation*}
\text{Spec}^*(\mathcal{B}) = \text{Spec}^*(T_u),
\end{equation*}
where $\text{Spec}^*(S)$ denotes the multiset of all eigenvalues of the operator $S$, with each eigenvalue repeating according to its multiplicity.\par
Denoting by $U_\varphi$ the natural embedding of $\End V_\varphi$ into $\prod_{\varphi \in \widehat{G}} \End V_\varphi$, we see that $U_\varphi$ is a $T_u$-invariant subspace, and under the identification of $U_\varphi$ with $\End V_\varphi$, the restriction $T_{u, \varphi} := \left. T_u \right|_{U_\varphi}$ is the multiplication operator by $\widehat{u}(\varphi)$, $T_{u, \varphi}(M) = \widehat{u}(\varphi) M$. By fixing a basis of $V_\varphi$ and viewing the elements of $\End V_\varphi$ as $d_\varphi \times d_\varphi$ matrices, we can decompose $\End V_\varphi$ into $T_{u, \varphi}$-invariant components as
\begin{equation*}
\End V_\varphi = (\End V_\varphi)^{(1)} \oplus ... \oplus (\End V_\varphi)^{(d_\varphi)},
\end{equation*}
where $(\End V_\varphi)^{(j)}$ denotes the set of $d_\varphi \times d_\varphi$ matrices whose columns are all zero except for column $j$. Denoting by $T_{u, \varphi}^{(j)}$ the restriction of $T_{u, \varphi}$ to $(\End V_\varphi)^{(j)}$, and identifying $(\End V_\varphi)^{(j)}$ with $V_\varphi$, we have $T_{u, \varphi}^{(j)}(v) = \widehat{u}(\varphi) v$, and therefore $\text{Spec}(T_{u, \varphi}^{(j)}) = \text{Spec}(\widehat{u}(\varphi))$.\\
We conclude that
\begin{equation*}
\text{Spec}^*\Big( \left. T_u \right|_{U_\varphi} \Big) = \bigsqcup_{j=1}^{d_\varphi} \text{Spec}^*\big(T_{u, \varphi}^{(j)}\big) = \bigsqcup_{j=1}^{d_\varphi} \text{Spec}^*\big(\widehat{u}(\varphi)\big),
\end{equation*}
and we finally obtain
\begin{equation} \label{eq:7}
\text{Spec}^*(\mathcal{B}) = \text{Spec}^*(T_u) = \bigsqcup_{\varphi \in \widehat{G}} \left(\, \bigsqcup_{j=1}^{d_\varphi} \text{Spec}^*\big(\widehat{u}(\varphi)\big) \right).
\end{equation}

\subsection{The Case of a Gelfand Pair}
We shall now focus our attention on the case where $(G, K)$ is a Gelfand pair and prove Theorem \ref{thm:1}.\par

\begin{proof}[Proof of Theorem \ref{thm:1}]
Let $\varphi \in \widehat{G}$ be such that $\widehat{u}(\varphi) \neq 0$. Since $(G, K)$ is a Gelfand pair, by Theorem 9 in \cite{Diaconis}, there is a basis $B$ of the representation space of $\varphi$ such that
\begin{equation*}
\Big[\widehat{\xi}(\varphi)\Big]_B = \text{diag}(b, 0, ..., 0)
\end{equation*}
for all bi-$K$-invariant functions $\xi$ on $G$, where, clearly, $b = \text{tr} \left( \widehat{\xi}(\varphi) \right)$.\\
Since $u$ is bi-$K$-invariant, it readily follows that
\begin{equation*}
\text{Spec}^*\big(\widehat{u}(\varphi)\big) = \{\text{tr} \left( \widehat{u}(\varphi) \right),  0, ..., 0 \}.
\end{equation*}
We can rewrite $\text{tr} \left( \widehat{u}(\varphi) \right)$ as follows:
\begin{equation*}
\begin{split}
\text{tr} \left( \widehat{u}(\varphi) \right) & = \text{tr} \int\limits_G u(x) \varphi(x^{-1}) \, d\mu(x) = \int\limits_G u(x) \chi_\varphi(x^{-1}) \, d\mu(x) \\
 & = \int\limits_G u(x) \overline{\chi_\varphi(x)} \, d\mu(x) = \left< u, \chi_\varphi \right>
\end{split}
\end{equation*}
where the character of $\varphi$ satisfies the identity $\chi_\varphi(x^{-1}) = \overline{\chi_\varphi(x)}$ since $\varphi$ is unitary. It follows that
\begin{equation} \label{eq:8}
\text{Spec}^*\big(\widehat{u}(\varphi)\big) = \{\left< u, \chi_\varphi \right>,  0, ..., 0 \}.
\end{equation}
Note that this is trivially also true when $\widehat{u}(\varphi) = 0$, and hence (\ref{eq:8}) is true for every $\varphi \in \widehat{G}$. By (\ref{eq:7}) it follows that
\begin{equation*}
\begin{split}
\text{Spec}^*(\mathcal{B}) &= \bigsqcup_{\varphi \in \widehat{G}} \left(\, \bigsqcup_{j=1}^{d_\varphi} \text{Spec}^*\big(\widehat{u}(\varphi)\big) \right) = \bigsqcup_{\varphi \in \widehat{G}} \left(\, \bigsqcup_{j=1}^{d_\varphi} \{\left< u, \chi_\varphi \right>,  0, ..., 0 \} \right) \\
 &= \bigsqcup_{\varphi \in \widehat{G}} \left(\, \bigsqcup_{j=1}^{d_\varphi} \{\left< u, \chi_\varphi \right> \} \right) \sqcup \left( \bigsqcup_{n = 1}^{\infty} \{0\} \right).
\end{split}
\end{equation*}
This precisely means that
\begin{equation*}
\mathrm{Spec}(\mathcal{B}) = \left\{ \left< u, \chi_\varphi \right> \mathrel{\big|} \varphi \in \widehat{G} \right\} \cup \{0\},
\end{equation*}
with each $\left< u, \chi_\varphi \right>$ appearing with multiplicity $\dim \varphi$, and $0$ having infinite multiplicity.
\end{proof}

\section{Case Study: Orbit POVMs on $S^2$} \label{sec:3}
In this section we restrict attention to the special case where $G = \text{SU}(2)$ and $K \simeq S^1$, and apply Theorem \ref{thm:1} to compute the spectrum of the Berezin transform in this case explicitly.\par
Let $G = \text{SU}(2)$ and fix a maximal torus of $G$,
\begin{equation*}
K = \left\{ k_t := \begin{pmatrix} e^{it} & 0\\0 & e^{-it} \end{pmatrix} \mathrel{\Big|} t \in [0, 2\pi) \right\} \simeq S^1.
\end{equation*}\par
Recall that the irreducible unitary representations of $\text{SU}(2)$ are given by $\widehat{G} = \left\{ \rho_j \mid j \in \frac{1}{2} \mathbb{N} \right\}$ where $\rho_j$ is a representation of dimension $n_j := 2j+1$ (Theorem 5.6.3 in \cite{Kowalski}), and for every such representation $(\rho_j, V_j)$ there is an orthonormal basis $w_j^{(j)}, w_j^{(j-1)}, ..., w_j^{(-j)}$ of $V_j$ consisting of eigenvectors w.r.t. a generator of $K$, $\rho_j(k_t) w_j^{(m)} = e^{2imt} w_j^{(m)}$. The parameter $m$ in $w_j^{(m)}$ is called the weight of the vector.\par
Fix $j \in \frac{1}{2} \mathbb{N}$ and take $v_{j, m} = w_j^{(m)}$ to be the vector of weight $m$. Then indeed
\begin{equation*}
K = \left\{ g \in G \mid \rho_j(g) v_{j, m} = e^{i\theta}v_{j, m} \text{ for some } \theta \in [0,2\pi) \right\}
\end{equation*}
is the "stabilizer" of $v_{j, m}$ and by direct computation or by recalling the Hopf fibration $S^3 \stackrel{S^1}\rightarrow S^2$, we find that the quotient space $\Omega = G/K$ is isomorphic to the sphere $S^2$. We thus have the POVM $dW_{j, m} = {n_j \, P_{j, m} \, d\omega}$ on $S^2$, where $P_{j, m}([g]) = P_{\rho_j(g)v_{j, m}}$ is the orthogonal projection onto $\rho_j(g)v_{j, m}$, with associated Berezin transform $\mathcal{B}_{j, m}$ and the corresponding function
\begin{equation*}
u_{j, m}(g) = n_j \left|\left< \rho_j(g) v_{j, m}, v_{j, m} \right>\right|^2.
\end{equation*}

\subsection{The Spectrum of an Orbit POVM on $S^2$} \label{sec:3.1}
We shall now prove Theorem \ref{thm:2}, which states that the positive eigenvalues of the Berezin transform $\mathcal{B}_{j, m}$ are $\lambda^{(0)}, \lambda^{(1)}, \hdots, \lambda^{(2j)}$, where
\begin{equation*}
\lambda^{(J)} := \frac{(2j)!(2j+1)!}{(2j-J)!(2j+J+1)!} \left( \sum\limits_{z=0}^{j-m} (-1)^z \frac{\binom{2j-J}{z}\binom{J}{j-m-z}^2}{\binom{2j}{j-m}} \right)^2
\end{equation*}
has multiplicity $2J+1$.

\begin{proof}[Proof of Theorem \ref{thm:2}]
We begin by making the following observation: The pair $(G, K)$ is a Gelfand pair. Indeed, one can readily check that for any $x = \begin{pmatrix} \alpha & -\bar{\beta} \\ \beta & \bar{\alpha} \end{pmatrix} \in \text{SU}(2)$, we have
\begin{equation*}
x = \begin{pmatrix} \alpha & -\bar{\beta} \\ \beta & \bar{\alpha} \end{pmatrix} = k_{(\phi-\pi)/2} \begin{pmatrix} \alpha & -\bar{\beta} \\ \beta & \bar{\alpha} \end{pmatrix}^{-1} k_{(\phi+\pi)/2} \in K x^{-1} K,
\end{equation*}
where $\phi = \text{arg}(\alpha)$. Hence, applying Proposition 6.1.3 from \cite{vanDijk} with $\theta(x) = x$, the desired conclusion follows.\par
Therefore, by Theorem \ref{thm:1}, it follows that
\begin{equation*}
\text{Spec}(\mathcal{B}_{j, m}) = \left\{ \left< u_{j, m}, \chi_\varphi \right> \mathrel{\big|} \varphi \in \widehat{G} \right\} \cup \{0\} = \left\{ \left< u_{j, m}, \chi_{\rho_J} \right> \mathrel{\Big|} J \in \frac{1}{2} \mathbb{N} \right\} \cup \{0\},
\end{equation*}
where $\lambda^{(J)} := \left< u_{j, m}, \chi_{\rho_J} \right>$ has multiplicity $\dim \rho_J = 2J+1$.\par
In order to calculate $\lambda^{(J)}$, we first rewrite $u_{j, m}$ as follows:
\begin{equation*}
\begin{split}
u_{j, m}(g) & = n_j \left|\left< \rho_j(g) v_{j, m}, v_{j, m} \right>\right|^2 \\
 & = n_j \left< \rho_j(g) v_{j, m}, v_{j, m} \right> \overline{\left< \rho_j(g) v_{j, m}, v_{j, m} \right>} \\
 & = n_j \left< \rho_j(g) v_{j, m}, v_{j, m} \right> \left< \overline{\rho_j(g)} \, \overline{v_{j, m}}, \overline{v_{j, m}} \right> \\
 & = n_j \left< \rho_j(g) w_j^{(m)}, w_j^{(m)} \right> \left< \overline{\rho_j}(g) \overline{w_j^{(m)}}, \overline{w_j^{(m)}} \right> \\
 & = n_j \left< \rho_j(g) w_j^{(m)}, w_j^{(m)} \right> \left< \rho_j(g) w_j^{(-m)}, w_j^{(-m)} \right> \\
 & = n_j \left< \left( \rho_j(g) \otimes \rho_j(g) \right) \left( w_j^{(m)} \otimes w_j^{(-m)} \right), w_j^{(m)} \otimes w_j^{(-m)} \right>,
\end{split}
\end{equation*}
where the equality $\left< \overline{\rho_j}(g) \overline{w_j^{(m)}}, \overline{w_j^{(m)}} \right> = \left< \rho_j(g) w_j^{(-m)}, w_j^{(-m)} \right>$ follows from the fact that $\overline{w_j^{(m)}}$ is the vector of weight $-m$ for the dual representation $\overline{\rho_j}$:
\begin{equation*}
\overline{\rho_j}(k_t) \overline{w_j^{(m)}} = \overline{\rho_j(k_t) w_j^{(m)}} = \overline{e^{2imt} w_j^{(m)}} = e^{-2imt} \overline{w_j^{(m)}}.
\end{equation*}\par
Let $r_j := \rho_j \otimes \rho_j$ and $y_{j, m} := w_j^{(m)} \otimes w_j^{(-m)}$. Then by our computation, $u_{j, m}(g) = n_j \left< r_j(g) y_{j, m}, y_{j, m} \right>$. By the well-known Clebsch-Gordan formula (Corollary 5.6.2 in \cite{Kowalski}), stating that for all $k \ge l \ge 0$,
\begin{equation*}
\rho_k \otimes \rho_l \simeq \bigoplus_{0 \le i \le l} \rho_{k+l-2i},
\end{equation*}
we conclude that in our setting,
\begin{equation} \label{eq:10}
\rho_j \otimes \rho_j \simeq \bigoplus_{J=0}^{2j} \rho_J.
\end{equation}
In particular, we can write
\begin{equation*}
y_{j, m} = \sum_{J=0}^{2j} y_{j, m}^{(J)}
\end{equation*}
for some $y_{j, m}^{(J)} \in V_J$. Then $y_{j, m}$ has total weight $0$ w.r.t. $r_j$:
\begin{equation*}
\begin{split}
r_j(k_t) y_{j, m} & = \left( \rho_j(k_t) w_j^{(m)} \right) \otimes \left( \rho_j(k_t) w_j^{(-m)} \right) = \left( e^{2imt} w_j^{(m)} \right) \otimes \left( e^{-2imt} w_j^{(-m)} \right) \\
 & = w_j^{(m)} \otimes w_j^{(-m)} = y_{j, m}.
 \end{split}
\end{equation*}
On the other hand, by (\ref{eq:10}),
\begin{equation*}
r_j(k_t) y_{j, m} = \sum_{J=0}^{2j} \rho_J(k_t) y_{j, m}^{(J)},
\end{equation*}
and hence $\rho_J(k_t) y_{j, m}^{(J)} = y_{j, m}^{(J)}$, which means all $y_{j, m}^{(J)}$ have weight $0$ as well, and we conclude that $y_{j, m}^{(J)} = \alpha_{j, m}^{(J)} \, w_J^{(0)}$ for some $\alpha_{j, m}^{(J)} \in \mathbb{C}$. We thus arrive at
\begin{equation*}
y_{j, m} = \sum_{J=0}^{2j} \alpha_{j, m}^{(J)} \, w_J^{(0)}.
\end{equation*}
Again by (\ref{eq:10}), we obtain
\begin{equation*}
\begin{split}
\left< r_j(g) y_{j, m}, y_{j, m} \right> & = \left< \sum_{J=0}^{2j} \alpha_{j, m}^{(J)} \, \rho_J(g) w_J^{(0)}, \sum_{J=0}^{2j} \alpha_{j, m}^{(J)} \, w_J^{(0)} \right> \\
 & = \sum_{J=0}^{2j} \left| \alpha_{j, m}^{(J)} \right|^2 \left< \rho_J(g) w_J^{(0)}, w_J^{(0)} \right>,
\end{split}
\end{equation*}
so that
\begin{equation*}
u_{j, m}(g) = n_j \sum_{J=0}^{2j} \left| \alpha_{j, m}^{(J)} \right|^2 \left< \rho_J(g) w_J^{(0)}, w_J^{(0)} \right>.
\end{equation*}\par
Now fix an irreducible representation $\rho_J$ for some $J \in \frac{1}{2}\mathbb{N}$. Using the basis $w_J^{(J)}, ..., w_J^{(-J)}$ for $V_J$, we conclude that its character is
\begin{equation*}
\chi_{\rho_J}(g) = \sum_{k=-J}^J \left< \rho_J(g) w_J^{(k)}, w_J^{(k)} \right>.
\end{equation*}
Therefore,
\begin{equation*}
\begin{split}
\left< u_{j, m}, \chi_{\rho_J} \right> & = \int\limits_G u_{j, m}(g) \overline{\chi_{\rho_J}(g)} \, d\mu(g) \\
 & = \sum_{k=-J}^J \sum_{J'=0}^{2j} n_j \left| \alpha_{j, m}^{(J')} \right|^2 \int\limits_G \left< \rho_{J'}(g) w_{J'}^{(0)}, w_{J'}^{(0)} \right> \overline{\left< \rho_J(g) w_J^{(k)}, w_J^{(k)} \right>} \, d\mu(g)
\end{split}
\end{equation*}
By Schur's orthogonality relations for matrix coefficients (Lemma 5.5.2 in \cite{Kowalski}),
\begin{equation*}
\int\limits_G \left< \rho_{J'}(g) w_{J'}^{(0)}, w_{J'}^{(0)} \right> \overline{\left< \rho_J(g) w_J^{(k)}, w_J^{(k)} \right>} \, d\mu(g)
\end{equation*}
vanishes for $J' \neq J$ and is equal to
\begin{equation*}
\frac{\left< w_J^{(0)}, w_J^{(k)} \right> \overline{\left< w_J^{(k)}, w_J^{(0)} \right>}}{\dim{V_J}} = \frac{\left< w_J^{(0)}, w_J^{(k)} \right>^2}{\dim{V_J}}
\end{equation*}
for $J'= J$. We conclude that $\left< u_{j, m}, \chi_{\rho_J} \right> = 0$ for $J \notin \{0, 1, 2, ..., 2j\}$, while for $J \in \{0, 1, 2, ..., 2j\}$ we have
\begin{equation*}
\begin{split}
\left< u_{j, m}, \chi_{\rho_J} \right> & = \sum_{k=-J}^J n_j \left| \alpha_{j, m}^{(J)} \right|^2 \int\limits_G \left< \rho_J(g) w_J^{(0)}, w_J^{(0)} \right> \overline{\left< \rho_J(g) w_J^{(k)}, w_J^{(k)} \right>} \, d\mu(g) \\
 & = \sum_{k=-J}^J n_j \left| \alpha_{j, m}^{(J)} \right|^2 \frac{\left< w_J^{(0)}, w_J^{(k)} \right>^2}{\dim{V_J}} = n_j \left| \alpha_{j, m}^{(J)} \right|^2 \frac{1}{\dim{V_J}}.
\end{split}
\end{equation*}
Finally, note that $\alpha_{j, m}^{(J)}$ is the Clebsch-Gordan coefficient $\left< j,m; j,-m \mid J,0 \right>$, and hence for $0 \le J \le 2j$ we find that
\begin{equation} \label{eq:11}
\begin{split}
\lambda^{(J)} & = \left< u_{j,m}, \chi_{\rho_J} \right> = \frac{2j+1}{2J+1} \left| \alpha_{j,m}^{(J)} \right|^2 = \frac{2j+1}{2J+1} \left| \left< j,m; j,-m \mid J,0 \right> \right|^2 \\
 & = \frac{2j+1}{2J+1} \cdot \frac{(2J+1)(2j-J)!J!^2}{(2j+J+1)!} \cdot (j+m)!^2 (j-m)!^2 J!^2 \\
 & \cdot \left( \sum_z \frac{(-1)^z}{z!(2j-J-z)!(j-m-z)!^2(J-j+m-z)!^2} \right)^2 \\
 & = \frac{(2j)!(2j+1)!}{(2j-J)!(2j+1+J)!} \left( \sum_{z=0}^{j-m} (-1)^z \frac{\binom{2j-J}{z}\binom{J}{j-m-z}^2}{\binom{2j}{j-m}} \right)^2,
\end{split}
\end{equation}
where the value of $\left< j,m; j,-m \mid J,0 \right>$ follows from the general formula for the Clebsch-Gordan coefficients \cite[p.~172]{Bohm}.
\end{proof}

In particular, we have
\begin{equation} \label{eq:12}
\begin{split}
\lambda^{(1)} & = \frac{(2j)!(2j+1)!}{(2j-1)!(2j+2)!} \left( \sum_{z=0}^{j-m} (-1)^z \frac{\binom{2j-1}{z}\binom{1}{j-m-z}^2}{\binom{2j}{j-m}} \right)^2 \\
 & = \frac{2j}{2j+2} \left( \frac{\binom{2j-1}{j-m} - \binom{2j-1}{j-m-1}}{\binom{2j}{j-m}} \right)^2 = \frac{j}{j+1} \left( \frac{(j+m)-(j-m)}{2j} \right)^2 \\
 & = \frac{m^2}{j(j+1)}.
\end{split}
\end{equation}

\subsection{Highest Weight Vector}
We first consider the case $m = j$ when $v_{j,m} = w_j^{(j)}$ is the highest weight vector.
Then (\ref{eq:11}) simplifies to
\begin{equation*}
\lambda^{(J)} = \frac{(2j)!(2j+1)!}{(2j-J)!(2j+1+J)!} = \frac{\binom{4j+1}{2j-J}}{\binom{4j+1}{2j}}
\end{equation*}
From the last expression it readily follows that in this case,
\begin{equation*}
1 = \lambda^{(0)} > \lambda^{(1)} > \hdots > \lambda^{(2j)} > 0,
\end{equation*}
and, in particular,
\begin{equation*}
\gamma(\mathcal{B}_{j,j}) = 1 - \lambda^{(1)} = 1 - \frac{j}{j+1} = \frac{1}{j+1} = \frac{1}{j} - \frac{1}{j(j+1)}.
\end{equation*}
Hence Conjectures \ref{conj:1} and \ref{conj:2} hold in this case (of $d = 0$).

\subsection{Second-Highest Weight} \label{sec:3.3}
Now we consider the case $m = j-1$ when $v_{j,m} = w_j^{(j-1)}$ is the second-highest weight vector. We shall prove Conjectures \ref{conj:1} and \ref{conj:2} in this case (of $d = 1$) as well.
\begin{proof}[Proof of Theorem \ref{thm:3}]
First, note that (\ref{eq:11}) simplifies to
\begin{equation*}
\begin{split}
\lambda^{(J)} & = \frac{(2j)!(2j+1)!}{(2j-J)!(2j+1+J)!} \left( \frac{J^2}{2j} - \frac{2j-J}{2j} \right)^2 \\
 & = (2j+1)(2j-1)!^2 \frac{(J^2+J-2j)^2}{(2j-J)!(2j+1+J)!}.
\end{split}
\end{equation*}
We begin by proving the first part of the theorem. We have, for $1 \le J \le 2j$,
\begin{equation*}
\frac{\lambda^{(J)}}{\lambda^{(J-1)}} = \frac{(J^2+J-2j)^2 (2j+1-J)}{(J^2-J-2j)^2 (2j+1+J)} = 1 - \frac{2J(J^4-(8j+1)J^2+12j^2+4j)}{(J^2-J-2j)^2 (2j+1+J)}
\end{equation*}
The sign of
\begin{equation*}
\frac{2J(J^4-(8j+1)J^2+12j^2+4j)}{(J^2-J-2j)^2 (2j+1+J)}
\end{equation*}
equals that of $J^4-(8j+1)J^2+12j^2+4j$, and we are thus led to investigate the domains of positivity and negativity of the latter. We have
\begin{equation*}
J^4-(8j+1)J^2+12j^2+4j = (J^2 - \kappa_-)(J^2 - \kappa_+),
\end{equation*}
where
\begin{equation*}
\kappa_\pm = \frac{8j+1 \pm \sqrt{16j^2+1}}{2}.
\end{equation*}
We conclude that $J^4-(8j+1)J^2+12j^2+4j$ is negative for $\kappa_- < J^2 < \kappa_+$, i.e. for $2j+1 \le J^2 \le 6j$ (since $J^2$ is integer), and positive for $J^2 \le 2j$ or $J^2 \ge 6j+1$.\\
It follows that $\lambda^{(J)} < \lambda^{(J-1)}$ for $J^2 \le 2j$, $\lambda^{(J)} > \lambda^{(J-1)}$ for $2j+1 \le J^2 \le 6j$ and then again $\lambda^{(J)} < \lambda^{(J-1)}$ for $J^2 \ge 6j+1$, which can be summarized as
\begin{equation*}
\lambda^{(0)} > \lambda^{(1)} > \hdots > \lambda^{(\lfloor\sqrt{2j}\rfloor)} < \hdots < \lambda^{(\lfloor\sqrt{6j}\rfloor)} > \hdots > \lambda^{(2j)},
\end{equation*}
in agreement with Figure \ref{fig:1a} and proving Conjecture \ref{conj:2} for $d = 1$.\par
We proceed to prove Conjecture \ref{conj:1}. From what we have shown already, we have
\begin{equation*}
\max(\lambda^{(1)}, \hdots, \lambda^{(2j)}) = \max(\lambda^{(1)}, \lambda^{(\lfloor\sqrt{6j}\rfloor)}),
\end{equation*}
and hence it is only left to prove that $\lambda^{(1)} > \lambda^{(\lfloor\sqrt{6j}\rfloor)}$.\\
First, we can estimate $\lambda^{(1)}$ from below for all $j \ge \frac{27}{2}$ as follows,
\begin{equation*}
\lambda^{(1)} = \frac{(j-1)^2}{j(j+1)} \ge \frac{\left(\nicefrac{25}{2}\right)^2}{\nicefrac{27}{2}\cdot\nicefrac{29}{2}} = \frac{625}{783},
\end{equation*}
since the expression for $\lambda^{(1)}$ is clearly monotone increasing in $j$.\\
We proceed to estimate $\lambda^{(\lfloor\sqrt{6j}\rfloor)}$ from above. We have
\begin{equation*}
\begin{split}
\lambda^{(\lfloor\sqrt{6j}\rfloor)} & = \frac{(2j) \cdot\hdots\cdot (2j+1-\lfloor\sqrt{6j}\rfloor)}{(2j+1+\lfloor\sqrt{6j}\rfloor) \cdot\hdots\cdot (2j+2)} \cdot \left(\frac{\lfloor\sqrt{6j}\rfloor^2+\lfloor\sqrt{6j}\rfloor-2j}{2j}\right)^2 \\
 & \le \left(1-\frac{\lfloor\sqrt{6j}\rfloor+1}{2j+1+\lfloor\sqrt{6j}\rfloor}\right) \cdot\hdots\cdot \left(1-\frac{\lfloor\sqrt{6j}\rfloor+1}{2j+2}\right) \cdot \left(\frac{4j+\sqrt{6j}}{2j}\right)^2 \\
 & \le \left(1-\frac{\lfloor\sqrt{6j}\rfloor+1}{2j+1+\lfloor\sqrt{6j}\rfloor}\right)^{\lfloor\sqrt{6j}\rfloor} \left(2+\sqrt{\frac{3}{2j}}\right)^2 \\
 & \le \exp\left(-\frac{\lfloor\sqrt{6j}\rfloor(\lfloor\sqrt{6j}\rfloor+1)}{2j+\lfloor\sqrt{6j}\rfloor+1}\right) \left(2+\sqrt{\frac{3}{2j}}\right)^2 \\
 & \le \exp\left(-\frac{(\sqrt{6j}-1)\sqrt{6j}}{2j+\sqrt{6j}+1}\right) \left(2+\sqrt{\frac{3}{2j}}\right)^2 \\
 & = \exp\left(-3\,\frac{1-\nicefrac{1}{\sqrt{6j}}}{1+\nicefrac{\sqrt{3}}{\sqrt{2j}}+\nicefrac{1}{2j}}\right) \left(2+\sqrt{\frac{3}{2j}}\right)^2
\end{split}
\end{equation*}
and notice that the last expression is monotone decreasing in $j$. Hence for all $j \ge \frac{27}{2}$ we have
\begin{equation*}
\begin{split}
\lambda^{(\lfloor\sqrt{6j}\rfloor)} & \le \exp\left(-3\,\frac{1-\nicefrac{1}{\sqrt{6j}}}{1+\nicefrac{\sqrt{3}}{\sqrt{2j}}+\nicefrac{1}{2j}}\right) \left(2+\sqrt{\frac{3}{2j}}\right)^2 \\
 & \le \exp\left(-3\,\frac{1-\nicefrac{1}{9}}{1+\nicefrac{1}{3}+\nicefrac{1}{27}}\right) \left(2+\frac{1}{3}\right)^2 = \frac{49}{9} \, e^{-\frac{72}{37}},
\end{split}
\end{equation*}
We conclude that for $j \ge \frac{27}{2}$,
\begin{equation*}
\lambda^{(\lfloor\sqrt{6j}\rfloor)} < \lambda^{(1)}.
\end{equation*}
Direct inspection shows that this inequality holds for all $\frac{5}{2} \le j < \frac{27}{2}$ as well. Therefore, for $j \ge \frac{5}{2}$,
\begin{equation*}
\lambda^{(1)} > \lambda^{(2)}, \lambda^{(3)}, \hdots, \lambda^{(2j)}
\end{equation*}
and consequently,
\begin{equation*}
\gamma(\mathcal{B}_{j,j-1}) = 1 - \lambda^{(1)} = 1 - \frac{(j-1)^2}{j(j+1)} = \frac{3}{j} - \frac{4}{j(j+1)},
\end{equation*}
proving Conjecture \ref{conj:1} for $d = 1$ and finishing the proof of the theorem.
\end{proof}

\subsection{Lower Weights} \label{sec:3.4}
Finally, we consider the case of unbounded $j - |m|$ and prove Theorem \ref{thm:4}, which asserts that in this case, the sequence of POVMs $\left(W_{j, m}\right)_{j \in \frac{1}{2} \mathbb{N}}$ does not yield a Berezin-Toeplitz quantization.
\begin{proof}[Proof of Theorem \ref{thm:4}]
Let $\left(m_j\right)_{j \in \frac{1}{2} \mathbb{N}}$ be any sequence of weights such that for all $j$, $m_j \in \{ -j, -j+1, \hdots, j \}$, and $(j-|m_j|)_{j \in \frac{1}{2} \mathbb{N}}$ is unbounded. Assume to the contrary that the sequence of maps $Q_\hbar(f) = \int_{S^2} f \, dW_{j, m_j}$ satisfies the properties of a quantization. By the definition of $Q_\hbar$, it is $\text{SU}(2)$-equivariant (as defined in \cite[p.~22]{IKP}). Indeed, for any $g \in \text{SU(2)}$ and $f \in C^\infty(S^2)$, we have the following straightforward calculation:
\begin{equation*}
\begin{split}
(g Q_\hbar)(f) &= \rho_j(g) Q_\hbar(f) \rho_j(g)^{-1} = \rho_j(g) \left( \int_{S^2} f \, dW_{j, m_j} \right) \rho_j(g)^{-1} \\
 &=  \rho_j(g) \left( \int_{S^2} f(x) \, n_j \, P_{\rho_j(x)v_{j,m_j}} \, d\omega(x) \right)  \rho_j(g)^{-1} \\
 &=  \int_{S^2} f(x) \, n_j \left( \rho_j(g) P_{\rho_j(x)v_{j,m_j}} \rho_j(g)^{-1} \right) d\omega(x) \\
 &= \int_{S^2} f(x) \, n_j \, P_{\rho_j(gx)v_{j,m_j}} \, d\omega(x) \\
 &= \int_{S^2} f(g^{-1} x) \, n_j \, P_{\rho_j(x)v_{j,m_j}} \, d\omega(x) \\
 &= \left( \int_{S^2} (g \cdot f) \, dW_{j, m_j} \right) = Q_\hbar (g \cdot f) = (Q_\hbar g)(f).
\end{split}
\end{equation*}\par
Therefore, by Theorem 6.2 of \cite{IKP}, it is equivalent (as defined in Definition 6.1 of \cite{IKP}) to $T_\hbar^{(t)}$ for some $t \ge 0$, where $T_\hbar^{(t)}(f) := T_\hbar(e^{-t\hbar\Delta}f)$. Here $T_\hbar$ is the standard quantization of $S^2$, given by $T_\hbar(f) = \int_{S^2} f \, dW_{j, j}$, where $j$ and $\hbar$ are connected via $\hbar = \frac{1}{2j}$. Explicitly, this means that there is a sequence of unitary operators $U_\hbar$ such that for every $f \in C^\infty(S^2)$ we have
\begin{equation} \label{eq:13}
\|U_\hbar T_\hbar^{(t)}(f) U_\hbar^{-1} - Q_\hbar(f)\|_\text{op} = O(\hbar^2)
\end{equation}\par
In order to exploit equation (\ref{eq:13}) to arrive at a contradiction, it will prove more convenient to work with the Berezin transforms rather than the quantization maps themselves, as we gained understanding of the spectra of the former. We have:
\begin{equation*}
\begin{split}
\mathcal{B}\Big(T_\hbar^{(t)}\Big) &= \frac{1}{n_\hbar} \left(T_\hbar^{(t)}\right)^* \left(T_\hbar^{(t)}\right) = \frac{1}{n_\hbar} \left(T_\hbar \circ e^{-t\hbar\Delta}\right)^* \left(T_\hbar \circ e^{-t\hbar\Delta}\right) \\
 &= \frac{1}{n_\hbar} e^{-t\hbar\Delta} \left(T_\hbar^{(t)}\right)^* \left(T_\hbar^{(t)}\right) e^{-t\hbar\Delta} = e^{-2t\hbar\Delta} \mathcal{B}(T_\hbar) = e^{-2t\hbar\Delta} \mathcal{B}_{j, j},
\end{split}
\end{equation*}
while, by definition,
\begin{equation*}
\mathcal{B}(Q_\hbar) = \mathcal{B}_{j, m_j}.
\end{equation*}\par
In order to make the transition from the quantization maps to the corresponding Berezin transforms, we shall pass from the operator norm to the trace norm. Towards this end, recall the general fact that if $V$ is a vector space of dimension $n$, then for any $A \in \End V$ we have $\|A\|_\text{op} \ge \frac{1}{\sqrt{n}} \|A\|_2$. Then, noting that $\|U_\hbar T_\hbar^{(t)}(f) U_\hbar^{-1}\|_\text{op} = \|T_\hbar^{(t)}(f)\|_\text{op} \le \| e^{-t\hbar\Delta} f \|_\infty \le \| f \|_\infty$ and $\|Q_\hbar(f)\|_\text{op} \le \| f \|_\infty$, we proceed to estimate
\begin{equation*}
\begin{split}
\|U_\hbar T_\hbar^{(t)}(f) U_\hbar^{-1} - Q_\hbar(f)\|_\text{op} & \ge \left| \|U_\hbar T_\hbar^{(t)}(f) U_\hbar^{-1}\|_\text{op} - \|Q_\hbar(f)\|_\text{op} \right| \\
 & = \left| \|T_\hbar^{(t)}(f)\|_\text{op} - \|Q_\hbar(f)\|_\text{op} \right| \\
 & \ge \frac{1}{2 \| f \|_\infty} \left| \|T_\hbar^{(t)}(f)\|_\text{op}^2 - \|Q_\hbar(f)\|_\text{op}^2 \right| \\
 & \ge \frac{1}{2 \| f \|_\infty n_\hbar} \left| \|T_\hbar^{(t)}(f)\|_2^2 - \|Q_\hbar(f)\|_2^2 \right| \\
 & = \frac{1}{2 \| f \|_\infty n_\hbar} \left| \left\langle \left(T_\hbar^{(t)}\right)^* \left(T_\hbar^{(t)}\right) f, f \right\rangle - \left\langle Q_\hbar^* Q_\hbar f, f \right\rangle \right| \\
 & = \frac{1}{2 \| f \|_\infty} \left| \left\langle \mathcal{B}\left(T_\hbar^{(t)}\right) f, f \right\rangle - \left\langle \mathcal{B}\left(Q_\hbar\right) f, f \right\rangle \right| \\
 & = \frac{1}{2 \| f \|_\infty} \left| \left\langle e^{-2t\hbar\Delta} \mathcal{B}_{j, j} f, f \right\rangle - \left\langle \mathcal{B}_{j, m_j} f, f \right\rangle \right|
\end{split}
\end{equation*}\par
Now consider the Laplace-Beltrami operator $\Delta$ on $S^2$ and let $\varphi_1$ be a normalized eigenfunction of $\Delta$ with eigenvalue $\lambda^{(1)}(\Delta) = 2$. We claim that $\varphi_1$ is also an eigenfunction of $\mathcal{B}_{j,m}$ with eigenvalue $\lambda^{(1)}(\mathcal{B}_{j,m})$ for every $m$. The reasoning is as follows. Consider the regular representation $\pi$ of $\text{SO(3)}$ on $L_2(S^2, \omega)$ given by $(\pi(g)f)(x) = f(g^{-1}x)$. By Proposition 6.4.2 in \cite{Kowalski}, it has a decomposition into a direct sum of irreducible representations
\begin{equation*}
L_2(S^2, \omega) = \bigoplus_{k=0}^{\infty} \mathcal{H}_k
\end{equation*}
where $\mathcal{H}_k$ is the space of the irreducible unitary representation $\pi_k$ of $\text{SO}(3)$, of dimension $\dim \mathcal{H}_k = 2k+1$. It is the space of spherical harmonics of degree $k$, i.e. harmonic homogeneous polynomials of degree $k$ in $3$ variables restricted to the sphere $S^2$. This is nothing but the eigenspace of $\Delta$ corresponding to the eigenvalue $\lambda_k(\Delta) = k(k+1)$. Now observe that $\mathcal{B}_{j, m}$ commutes with $\pi$:
\begin{equation*}
\begin{split}
\mathcal{B}_{j,m}(\pi(g)f)(s) &= \int_\Omega u_{j,m}(t^{-1}s) \, (\pi(g)f)(t) \, dt = \int_\Omega u_{j,m}(t^{-1}s) \, f(g^{-1}t) \, dt \\
 &= \int_\Omega u_{j,m}(t^{-1}g^{-1}s) \, f(t) \, dt = \mathcal{B}_{j,m}(f)(g^{-1}s) \\
 &= (\pi(g)\mathcal{B}_{j,m}(f))(s).
\end{split}
\end{equation*}
Hence $\mathcal{B}_{j,m}$ is an intertwiner between every pair of irreducible representations $\mathcal{H}_k$ and $\mathcal{H}_l$. Since $\dim{\mathcal{H}_k} = 2k+1$, all $\mathcal{H}_k$ are distinct, and thus from Schur's lemma it follows that we can write
\begin{equation*}
\mathcal{B}_{j,m} = \sum_{k=0}^{\infty} \mu_k I_{\mathcal{H}_k}
\end{equation*}
for some real constants $\mu_k$. On the other hand, we know that
\begin{equation*}
\mathcal{B}_{j,m} = \sum_{k=0}^{2j} \lambda^{(k)}(\mathcal{B}_{j,m}) I_{E_k}
\end{equation*}
where $E_k$ is the eigenspace of $\mathcal{B}_{j,m}$ corresponding to the eigenvalue $\lambda^{(k)}(\mathcal{B}_{j,m})$, which by Theorem \ref{thm:2} has dimension $2k+1$. From this it readily follows that $E_k = \mathcal{H}_k$ and $\mu_k = \lambda^{(k)}(\mathcal{B}_{j,m})$ for all $0 \le k \le 2j$. Hence we have the decomposition
\begin{equation*}
\mathcal{B}_{j,m} = \sum_{k=0}^{2j} \lambda^{(k)}(\mathcal{B}_{j,m}) I_{\mathcal{H}_k}
\end{equation*}
and our claim follows.
\par
In particular, we conclude that $\mathcal{B}_{j, m_j} \, \varphi_1 = \lambda^{(1)}(\mathcal{B}_{j, m_j}) \varphi_1 = \frac{j}{j+1} \left( \frac{m_j}{j} \right)^2 \varphi_1$ and $\mathcal{B}_{j, j} \, \varphi_1 = \lambda^{(1)}(\mathcal{B}_{j, j}) \varphi_1 = \frac{j}{j+1} \varphi_1$ (cf. (\ref{eq:12})). Then
\begin{equation*}
e^{-2t\hbar\Delta} \mathcal{B}_{j, j} \, \varphi_1 = e^{-4t\hbar} \frac{j}{j+1} \varphi_1,
\end{equation*}
and it follows that
\begin{equation*}
\begin{split}
O(\hbar^2) &= 2 \| \varphi_1 \|_\infty \|U_\hbar T_\hbar^{(t)}(\varphi_1) U_\hbar^{-1} - Q_\hbar(\varphi_1)\|_\text{op} \\
 & \ge \left| \left\langle e^{-2t\hbar\Delta} \mathcal{B}_{j, j} \, \varphi_1, \varphi_1 \right\rangle - \left\langle \mathcal{B}_{j, m_j} \, \varphi_1, \varphi_1 \right\rangle \right| \\
 & = \left| \left\langle e^{-4t\hbar} \frac{j}{j+1} \varphi_1, \varphi_1 \right\rangle - \left\langle \frac{j}{j+1} \left( \frac{m_j}{j} \right)^2 \varphi_1, \varphi_1 \right\rangle \right| \\
 & = \frac{j}{j+1} \left| e^{-4t\hbar} - \left( \frac{m_j}{j} \right)^2 \right| \\
 &= \frac{j}{j+1} \left| \left( 1 - 4t\hbar + O(\hbar^2) \right) - \left( 1 - \frac{j^2 - m_j^2}{j^2} \right) \right| \\
 &= (1-O(\hbar)) \left| \frac{(j-|m_j|)(j+|m_j|)}{j^2} - 4t\hbar + O(\hbar^2) \right|.
\end{split}
\end{equation*}
We conclude that
\begin{equation*}
\left| \frac{(j-|m_j|)(j+|m_j|)}{j^2} - 4t\hbar \right| = O(\hbar^2),
\end{equation*}
and, in particular,
\begin{equation*}
\frac{(j-|m_j|)(j+|m_j|)}{j^2} = O(\hbar).
\end{equation*}
However, we have
\begin{equation*}
1 \le \frac{j+|m_j|}{j} \le 2,
\end{equation*}
so we remain with
\begin{equation*}
\frac{j-|m_j|}{j} = O(\hbar),
\end{equation*}
or in other words,
\begin{equation*}
j-|m_j| = O(1),
\end{equation*}
which is a contradiction to our assumptions.
\end{proof}

\begin{remark}
Consider the complementary case where $j-|m|$ is bounded, and specifically assume that $j-m = d$ is constant. Let $Q_\hbar^{(d)}(f) := \int_{S^2} f \, dW_{j,j-d}$. Then a comparison of the spectra of the Berezin transforms of $Q_\hbar^{(d)}$ and $T_\hbar^{(t)}$ does not lead to a contradiction as in the proof above. Specifically, there exists $t \ge 0$ (in fact, $t = d$) such that for every fixed positive integer $k$, we have
\begin{equation*}
\left| \lambda^{(k)} \bigg( \mathcal{B} \Big( Q_\hbar^{(d)} \Big) \bigg) - \lambda^{(k)} \bigg( \mathcal{B} \Big( T_\hbar^{(t)} \Big) \bigg) \right| = O(\hbar^2)
\end{equation*}
as $\hbar \rightarrow 0$. One can verify this fact by showing by direct computation using (\ref{eq:11}) that both $\lambda^{(k)}$ values above are equal to $1 - k(k+1)(2d+1)\hbar$ up to $O(\hbar^2)$.\\
The question whether $Q_\hbar^{(d)}$ is a quantization of $S^2$ remains open for all $d \geq 1$.
\end{remark}

\section*{Acknowledgments}
This paper is based on my M.Sc. thesis at Tel Aviv University under the supervision of Prof. Leonid Polterovich.\par
I would like to express my deepest gratitude to Prof. Leonid Polterovich, whose experience and guidance have enriched me deeply. His insights and advice were indispensable for this work.\par
I would like to thank Prof. Andre Reznikov for helpful discussions that were crucial for the present work.\par
I would also like to thank Louis Ioos for dedicating his time to review this work and provide helpful comments.\par
Finally, I would like to thank Stéphane Nonnenmacher. The interest in the problem of whether other weight vectors yield a quantization was triggered by a question of his.

\end{document}